\documentclass[11pt,a4paper,final]{article}
\usepackage[left=1in, right=1in, top=0.5in, bottom=0.5in]{geometry}
\usepackage{enumitem} 
\usepackage{array}
\usepackage{booktabs}

\usepackage{times}
\usepackage{natbib}
\usepackage{lineno,hyperref}
\usepackage[utf8]{inputenc}
\usepackage[T1]{fontenc}
\usepackage[english]{babel}
\usepackage{amsmath}
\usepackage{amsfonts}
\usepackage{amssymb}
\usepackage{makeidx}
\usepackage{tfrupee}
\usepackage{tikz}
\usepackage{longtable}
\usepackage{graphicx}
\usepackage{ifvtex}
\usepackage{filecontents}
\usepackage{tikz}
\usepackage{float}
\usetikzlibrary{patterns}
\usepackage{pgfplots}
\usepackage{algorithm}
\usepackage{algorithmicx}
\usepackage{algpseudocode}
\usepackage{adjustbox}
\pgfplotsset{compat=1.17}
\usetikzlibrary{backgrounds}
\usetikzlibrary{arrows}
\usetikzlibrary{shapes,shapes.geometric,shapes.misc}
\pgfdeclarelayer{edgelayer}
\pgfdeclarelayer{nodelayer}
\pgfsetlayers{background,edgelayer,nodelayer,main}
\tikzstyle{none}=[inner sep=0mm]
\tikzset{new style 0/.style={circle,draw}}
\newtheorem{theorem}{Theorem}
\newtheorem{corollary}{Corollary}
\newtheorem{lemma}{Lemma}
\newtheorem{proposition}{Proposition}

\newtheorem{definition}{Definition}
\newtheorem{remark}{Remark}
\newenvironment{proof}{\paragraph{Proof:}}{\hfill$\square$}
\title{Strategic Bid Shading in Real-Time Bidding Auctions in Ad Exchange Using Minority Game Theory}
\author{Dipankar Das\\ Assistant Professor\\	Goa Institute of Management\\	Email ID:dipankar@gim.ac.in;dipankar3das@gmail.com
}
\date{}
\begin{document}
\maketitle
\begin{abstract}
Traditional auction theory posits that bid value exhibits a positive correlation with the probability of securing the auctioned object in ascending auctions. However, under uncertainty and incomplete information, as is characteristic in real-time advertising markets, truthful bidding may not always represent a dominant strategy or yield a Pure Strategy Nash Equilibrium. Real-Time Bidding (RTB) platforms operationalize impression-level auctions via programmatic interfaces, where advertisers compete in first-price auction settings and often resort to bid shading, i.e., strategically submitting bids below their private valuations to optimize payoff.\\	
This paper empirically investigates bid shading behaviors and strategic adaptation using large-scale RTB auction data from the Yahoo Webscope dataset. Integrating Minority Game Theory with clustering algorithms and variance-scaling diagnostics, we analyze equilibrium bidding behavior across temporally segmented impression markets. Our results reveal the emergence of minority-based bidding strategies, wherein agents partition hourly ad slots into submarkets and place bids strategically where they anticipate being in the numerical minority. This strategic heterogeneity facilitates reduced expenditure while enhancing win probability, functioning as an endogenous bid shading mechanism. The analysis highlights the computational and economic implications of minority strategies in shaping bidder dynamics and pricing outcomes in decentralized, high-frequency auction environments.
\end{abstract}
\pagebreak
\begin{center}
	\textbf{Strategic Bid Shading in Real-Time Bidding Auctions in Ad Exchange Using Minority Game Theory}
\end{center}
\section{Introduction}\label{sec:Intro}
Real-Time Bidding (RTB) represents a computational marketplace architecture that parallels electronic stock exchanges, enabling the algorithmic purchase and sale of digital advertising slots in real time. The process, from a user visiting a website to a bid request being sent and the display of the winning advertisement, typically occurs within 100 milliseconds. RTB platforms implement impression-level auctions, allocating advertising slots through automated bidding agents deployed by advertisers on Demand Side Platforms (DSPs). These agents respond dynamically to bid requests issued by Supply Side Platforms (SSPs) or Ad Exchanges (ADXs), factoring in budget constraints, user-level contextual features, and campaign-specific objectives.\\
Many companies/firms compete in the DSP to get the desired slots. On the other hand, by default there are many digital slots available in the form of mobile applications, websites, gaming, and other platforms, etc. The real viewers of these Ads are the real users of these applications, websites, and platforms. \\
\indent Within this ecosystem, user conversions are categorized into two primary types: post-click conversions, which occur following a user click-through and uninterrupted interaction on the advertiser’s website; and post-view conversions, where users convert without direct interaction but through latent impression influence. Advertisers register campaigns with DSPs to participate in RTB, specifying budget allocations and targeting parameters. As per the targets, the algorithm is designed to achieve the proposed targets throughout the period for which the bidder is willing to bid for digital slots. In addition to this, bidding happens dynamically. Bidding agents operational on the DSPs compete in impression-level auctions using first-price (or generalized first-price) mechanisms. The winning bidder’s advertisement is subsequently delivered, and DSPs track downstream user behavior (e.g., click-through rate, conversion rate) to refine future bidding strategies and optimize campaign-level Key Performance Indicators (KPIs), such as cost-per-click (CPC) and cost-per-acquisition (CPA) \cite{mahdian2008pay}.\\
A set of willing bidders fails to compete for the digital slots, and the current article does not consider them.\\
\indent RTB facilitates granular targeting using real-time user-level data, aligning advertising decisions with principles of personalized information retrieval and adaptive decision-making. Advertisers aim to acquire impressions that maximize expected utility concerning KPIs. At the same time, publishers allocate impressions to advertisers that yield the highest revenue, a dual-sided optimization problem grounded in game theory and dynamic pricing. The confluence of behavioral tracking, impression-level auctions, and budget-aware bidding transforms RTB into a computational analog of high-frequency financial markets, where algorithmic agents engage in strategic interactions under uncertainty.\\
But the optimization model is reference-dependent. This means the advertisers prefer a particular slot, and no such demand may exist for that slot. Hence, the winning slots may not generate the optimized revenue. This is the point that the paper considers when determining the advertisers' bidding strategy. The standard mechanism design under auction theory does not consider this phenomenon. \\
\indent Despite advances in auction modeling and bidder behavior prediction, understanding equilibrium bidding strategies in high-frequency, data-rich environments remains challenging. Traditional auction theory posits that the probability of winning increases monotonically with bid value in ascending auctions. However, truthful bidding may not be optimal under incomplete information and outcome uncertainty, which is common in RTB markets. Advertisers may engage in \emph{bid shading}, wherein they deliberately understate their valuations to reduce cost while maintaining competitive viability.\\
The DSP expects that the bidders are truthful about all slots, but this is not the case. We use this as the bid shading.\\ 
\indent This paper investigates bid shading behavior through the lens of Minority Game (MG) Theory, a computational model of strategic adaptation among boundedly rational agents. Using the Yahoo Webscope RTB dataset, we identify systematic patterns where bidding agents segment hourly impression markets and strategically act as minorities to evade direct price competition. This emergent behavior serves as an endogenous mechanism of bid shading, influenced by informational asymmetries and the decentralized structure of the auction environment.\\
\indent The remainder of this article is structured as follows: Section 2 reviews related literature and outlines the study’s objectives. Section 3 describes the Yahoo Webscope dataset used for empirical analysis. Section 4 presents the theoretical model, including a formalization of the Minority Game framework in the context of RTB auctions. Section 5 formalizes the findings regarding propositions, theorems, lemmas, and corollaries to interpret the issue in a general model framework. Section 6 reports the empirical findings and validates the theoretical insights through data-driven simulations. Finally, Section 6 concludes the study with implications for future research in computational auction design and algorithmic market behavior.
\section{Literature Review and Objective of the Study}
The present article is based on two types of literature: auction and minority game theory. We have merged these two to get some meaningful insight. Therefore, in the first part, we have presented a set of literature related to real-time bidding and bidding behavior. In the second part, we have presented the literature on minority game theory. 
\paragraph{Auction Mechanism and Bid Behavior} In sponsored search, several advertising slots are available on a search results page. They must be allocated among advertisers competing to advertise on the page. This situation gives rise to a bipartite matching market that is typically cleared by means of an automated auction mechanism \cite{aggarwal2009general}. The model is sufficiently flexible to accommodate bidder and position-specific minimum and maximum prices and different values for different slots. This suggests that in real-time bidding (RTB), the bidders have preferences regarding ad position and are heterogeneous in nature. The bid also depends on the probability that the selected ad-impression will be converted into a click. The standard matching algorithm stipulates that each matched pair of advertiser and slot has a positive price, while each unmatched slot has a zero cost. Therefore, selecting a low-priced slot with the possibility of higher ad impressions is the target of the bidder as a strategy. 
A widely adopted assumption in optimal auction design research is that the private values of advertisers follow log normal distribution\cite{edelman2006optimal,myerson1981optimal,ostrovsky2023reserve,xiao2009optimal}. Bids are then generated from the private values. However, a research dataset found that the 1st highest bid did not show strong log-normal distribution properties, nor did the 1st and 2nd highest bids, or all bids \cite{yuan2013real}. Moreover, it is found that bid amounts influence whether an impression is won and displayed, with higher bids increasing the likelihood of securing impressions, particularly during periods of high demand or when floor prices are involved.
The daily pacing refers to the way that advertisers spend
their budget in a single day. Analysis of datasets acquired from a production ad exchange revealed that floor price detection, daily pacing, and frequency setting are problems that are not addressed. Under certain dependency assumptions, a simple bidding function has been derived that can be calculated in real time; findings show that the optimal bid has a non-linear relationship with the impression level evaluation, such as the click-through rate and the conversion rate, which are estimated in real time from the
impression level features \cite{zhang2014optimal}. A  mathematical model suggests that optimal bidding strategies should try to bid for more impressions rather than focus on a small set of high-value impressions. An article by \cite{zhang2014real} suggests that data mining work, particularly bidding strategy development, becomes crucial in this performance-driven business. Hence, thus far, it is clear that the advertiser is at risk of not getting a slot or clicks from a specific ad slot. To minimize this risk, alternative options have been proposed in \cite{chen2015lattice,chen2015multi}. An ad option that gives an advertiser the right but not the obligation to purchase the future impressions or clicks from a specific ad slot (or user tag or keyword) at a pre-specified price. In \cite{cai2017real, du2019infer} it is established that bid decision is not a static optimization problem of either treating the value of each impression independently or setting a bid price to each segment of ad volume.  It is exciting to devise an optimal bidding strategy sequentially so that the entire ad campaign budget can be dynamically allocated across all the available impressions based on immediate and future rewards. A paper explains the relationship between bids and ad impressions through "cycling" in first-price auctions, where advertisers continuously adjust their bids to optimize placement and costs. This strategic behavior leads to dynamic patterns in bid adjustments, directly influencing which ads secure impressions and the corresponding auction revenues \cite{edelman2007strategic}. To address the complexity and dynamic nature of the RTB process, an auto pricing strategy (APS) approach has been proposed to determine the applications to bid for and their respective bid prices from the advertising agencies' perspective. \cite{adikari2019new}. An interesting finding is that RTB auctions are known to have relatively few bids per impression. Second, since RTB allows for cherry-picking of impressions in real time, it negatively aﬀects the quality of impressions assigned to reservation contracts. As such, advertisers' willingness to pay for impressions in reservation contracts
declines, lowering publishers' revenue from reservation contracts. This is because the rapid growth of RTB has created new challenges for advertisers and publishers on how much budget and ad inventory to allocate to RTB \cite{sayedi2018real}. A new advertiser always bids higher (sometimes above
valuation) in the beginning. The incumbent advertiser’s strategy depends on its valuation and CTR. A strong incumbent increases its bid to deter the publisher from learning the new advertiser’s CTR, whereas a weak incumbent decreases its bid to facilitate learning \cite{choi2019learning}. An article \cite{kinnear2020optimal} finds that an important characteristic is that bidding does not involve item valuations. The demand side platform seeks to fulfill acquisition contracts, not maximize the value of items won. The bidder or the advertiser sets utility maximization with a budget constraint as the primary goal in real-time bidding (RTB) systems \cite{ghosh2020optimal}. Therefore, if the utility function depends on the ad position, the ad's time, and the viewers' audiences, then this optimization problem will not allow them to choose a bid point that maximizes the ad-impression. This is so because ad impressions are heterogeneous regarding their positions to display. The paper \cite{kinnear2021real} explains the relationship between bid and ad impressions using the concept of supply rate curves, which quantify the expected number of items (or impressions) won at a given bid level. It highlights that higher bids generally increase the likelihood of winning impressions. Still, the cost and probability of winning vary based on the auction type (first-price or second-price) and the bid competition landscape. A paper discusses the relationship between bid and ad impressions by explaining that the bid amount directly impacts the likelihood of winning impressions in an auction. Specifically, higher bids are more likely to secure impressions. Still, this relationship is influenced by factors such as budget constraints, the dynamics of the auction environment, and the design of the bidding strategy, which optimizes clicks or impressions within a given budget \cite{liu2022real}. Another article explains the relationship between bid and ad impressions by modeling the optimal bid price as a function of the predicted key performance indicator (KPI, such as CTR) and the winning probability. It highlights that higher bid prices increase winning probabilities for impressions, but this must be balanced against budget constraints to maximize total impression value for an advertising campaign \cite{lu2022functional}.
\paragraph{Minority Game Theory} Consider a body of literature on the minority game and its applications in exchange markets. The Minority Game (MG) \cite{challet1997emergence} represents a mathematical interpretation of Brian Arthur's " El Farol Bar" problem \cite{arthur1994inductive}. This concept examines how individuals can collectively resolve an issue by adjusting their predictions about future outcomes. The primary focus of this problem is to investigate the process by which multiple individuals can reach a shared solution by modifying their personal expectations regarding upcoming events. This represents an extension of multi-agent systems. The dynamic nature of the problem raises new questions, such as whether individuals can reach a steady state and under which circumstances this steady state is stable \cite{moro2004minority}. The problem of the " El Farol Bar " investigates how one models "Bounded Rationality" in economic issues like exchange-traded markets. The problem has been modeled using a non-cooperative coordination game in which negative externalities determine payoffs. This is a repeated market-entry game with bounded rational agents \cite{whitehead2008farol}. Therefore, these attributes match the present issues with the RTB market. The mathematical structure of the minority game has been examined in \cite{challet2004mathematical}. Applications of the minority game in financial markets, particularly in stock exchange markets, have been explored in \cite{challet2004minority}. The minority game can generate synthetic price data to enhance market understanding. The 'synthetic' price history employed in these simulations, using data from real financial time series, leads to the notable finding that agents can collectively learn to identify moments in the market where profit is attainable. Consequently, when applied to real financial data, the system as a whole can perform better than random \cite{jeffries2003market}. Standard economic theory assumes a few agents compete in an uncooperative game, such as an auction. However, the RTB market resembles a stock exchange where many agents interact. Thus, the standard assumption of rationality is not realistically applicable. Therefore, a theory supporting evolutionary concepts, where agents make decisions, commit errors, learn, and adjust accordingly, was necessary. Physicists tend to view a game with a large number of players as a statistical system, necessitating the exploration of new approaches in which emerging collective phenomena can be better appreciated, and the minority game is one of the most widely accepted methods \cite{CHALLET1997407}. The market continually seeks new adaptive strategies for survival. The minority game is one method to identify such strategies. This is because the stock exchange is subject to mean reversion and represents a zero-sum game. As highly successful investors' strategies are revealed, many traders tend to adopt them. When most market participants begin using a strategy, its profitability gradually diminishes. Only early adopters benefit from it. The minority game can be employed to identify the minority strategy regarding time \cite{ZAPART20091157}. An important concept in relation to this is the evolutionary minority game theory, and a set of articles on this topic can be found in \cite{BURGOS2005518, weibull1997evolutionary}.
\paragraph{Research Gaps and Objective of the Study}
\indent The primary research gaps are as follows: (i) the extant literature has not addressed the bidding behavior pattern under the first-price auction mechanism, (ii) bid shading, i.e., given a private value for impressions, how should a bidder adjust their bid in a first-price auction? Bid shading in a first-price auction involves submitting a bid lower than the bidder's private value for an impression while maintaining a high probability of winning. The objective of the bidder is to maximize profits by reducing payment while still securing impressions. Utilizing existing methodology to calculate Bid Shading based on the dataset provided by Yahoo is not feasible, as the private value is unknown. Consequently, this study proposes a novel bidding strategy employing the concept of minority game theory to address this issue.
\section{Dataset Description}
\label{sec:data}
This study utilizes a proprietary dataset released by Yahoo as part of the Yahoo Research Webscope program \cite{yahoo_auction_2020}.\footnote{Brendan Kitts (2020), Auction State for a Sample of Real-Time Bid Video Ads Version 1.0, Yahoo Research Webscope dataset ydata-o\&o-video-auction-landscapes-2018-v1\_0, https://webscope.sandbox.yahoo.com/, June 20, 2020.} The dataset comprises a curated sample of video advertisements served via Verizon’s Supply-Side Platform (SSP) during 2018, and the associated auction-level data from Real-Time Bidding (RTB) markets.\\
\indent Each auction record is timestamped at the granularity of date and hour and is accompanied by a vector of discretized bid price intervals ranging from \$0 to \$50 in increments of \$0.10. For each interval, the dataset provides the corresponding number of impressions that could have been acquired at that price point. The dataset thus enables the reconstruction of granular bid landscapes across time and serves as a fertile ground for modeling auction market dynamics.\\
In particular, the dataset facilitates research in:
\begin{itemize}
	\item \textbf{Auction Forecasting:} Predicting future auction landscapes based on historical clearing prices and bid distributions across 1–14 day horizons.
	\item \textbf{Forward-Looking Bidding Algorithms:} Designing adaptive bidding strategies that optimize impression acquisition under predicted future auction states.
	\item \textbf{Bid Shading Optimization:} Evaluating how much a bidder should shade their bid in a first-price auction context, conditioned on private valuations and win-probability metrics over time.
\end{itemize}
Given its temporal structure and bid granularity, this dataset provides an ideal test-bed for exploring strategic behavior under bounded rationality, focusing on minority-based bidding strategies in high-frequency auction environments.
\subsection{Research Hypotheses}
\label{sec:hypothesis}
The following hypotheses are formulated to investigate strategic bidding behavior in RTB environments, specifically under the theoretical framework of Minority Game dynamics:
\begin{enumerate}[label=H\arabic*]
	\item \textit{Advertisers (i.e., bidders) engaged in RTB auctions adopt a bidding strategy consistent with the Minority Game paradigm.}
	\item \textit{Minority Strategy-based bidding behavior exhibits temporal consistency, persisting across multiple time resolutions (hourly, daily, weekly, monthly, and annually).}
	\item \textit{Bidders participating in minority strategies exhibit segmentation based on bid magnitude, resulting in differentiated behavioral clusters.}
	\item \textit{The bid distribution among minority participants is right-skewed toward lower bid values, particularly in auction intervals with fewer available impressions, indicating strategic cost minimization.}
	\item \textit{Competitive bidding intensifies during low-impression periods, reinforcing the viability of minority-based strategies in sparse auction contexts.}
\end{enumerate}
These hypotheses are empirically tested using statistical inference techniques, clustering analysis, Minority Game modeling, and variance-scaling diagnostics. Specifically, hypotheses \textbf{H1} to \textbf{H3} are validated through simulation and empirical analysis in Sections~\ref{sec:theory} and~\ref{sec:results}. Hypothesis \textbf{H4} is evaluated in Subsection~5.2, while Hypothesis \textbf{H5} is assessed in Subsection~5.3.
\section{Theoretical Framework}
\label{sec:theory}
This section develops the theoretical foundation for analyzing strategic bidding behavior in Real-Time Bidding (RTB) environments, focusing on the emergence of minority strategies in impression-level auctions.\\
\indent An initial exploratory analysis of the relationship between bid values and hourly impressions (from data, it is represented as a variable \texttt{imps\_day}) reveals a non-linear, bid-sensitive functional dependency. While marginal increases in bid (in increments of 0.10) occasionally result in disproportionately higher impressions, this relationship is neither stable nor deterministic. Accordingly, the identification of a bid path that maximizes impressions requires analyzing dynamic interactions across time scales and among heterogeneous agents.\\
\indent Figures~1 through~6 illustrate this relationship across six temporal resolutions: hourly, daily, weekly, monthly, yearly, and aggregated all-time, demonstrating that no universal bid-impression mapping exists. Bidders exhibit significant preference, valuation, and strategy heterogeneity, yielding divergent impressions at identical bid levels \cite{ghoshal2023serving}. Notably, higher bids do not consistently yield greater impression volumes. Instead, temporal aggregation (especially over yearly windows) reveals more structured, albeit bidder-specific, functional relationships. These observations suggest strategic temporal targeting, where advertisers defer aggressive bidding until optimal conditions emerge, leading to patterned bidding behavior. According to the paper \cite{grigasexpress}, the advertiser’s budget utilization proxy function is concave, signifying that the bidders are risk-averse. This supports the minority strategy hypothesis as proposed in the present article.
\begin{figure}[H]
	\centering
	\begin{minipage}[b]{0.32\textwidth}
		\includegraphics[width=\textwidth]{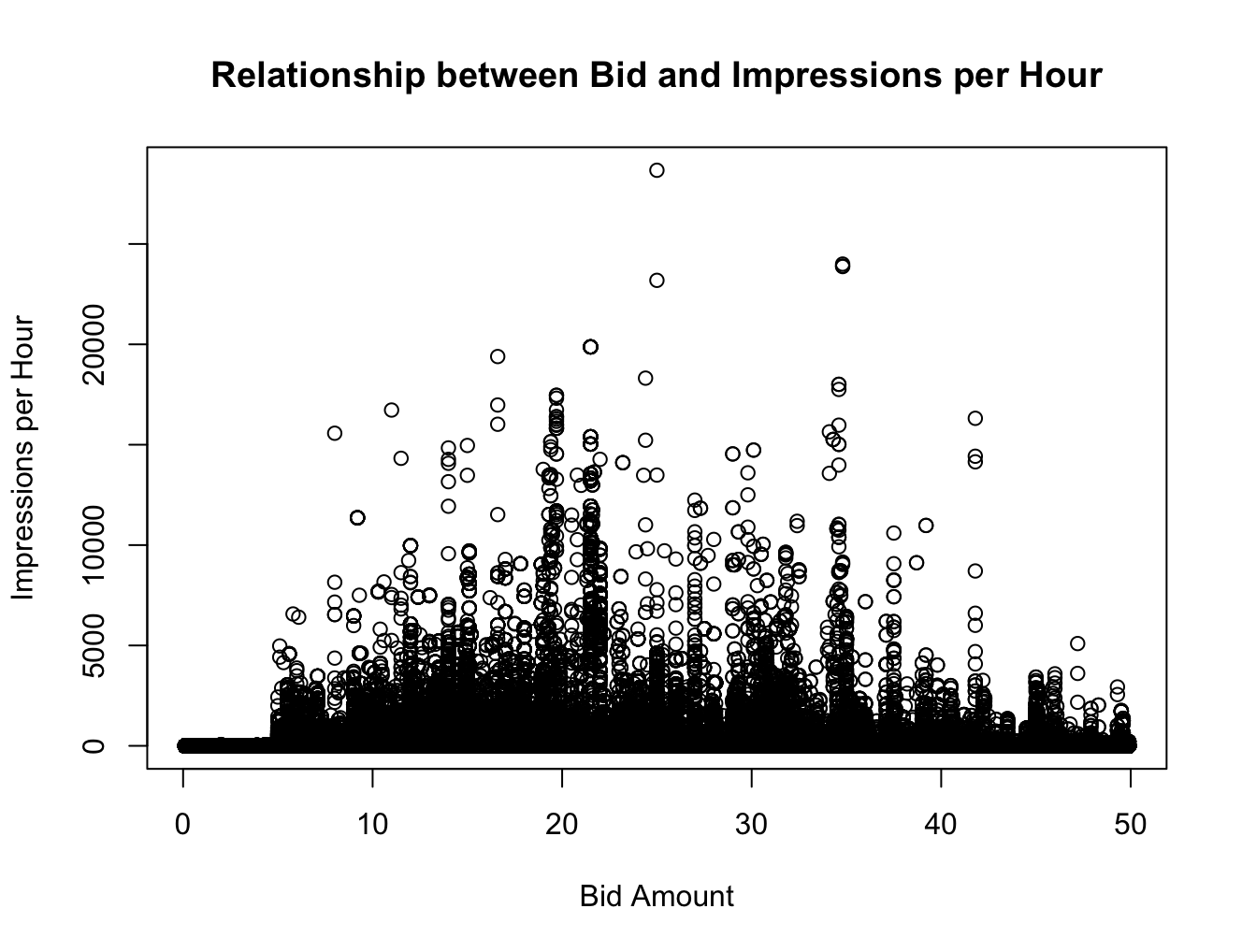}
		\caption{Bid vs. Impressions per Hour}
	\end{minipage}
	\hfill
	\begin{minipage}[b]{0.32\textwidth}
		\includegraphics[width=\textwidth]{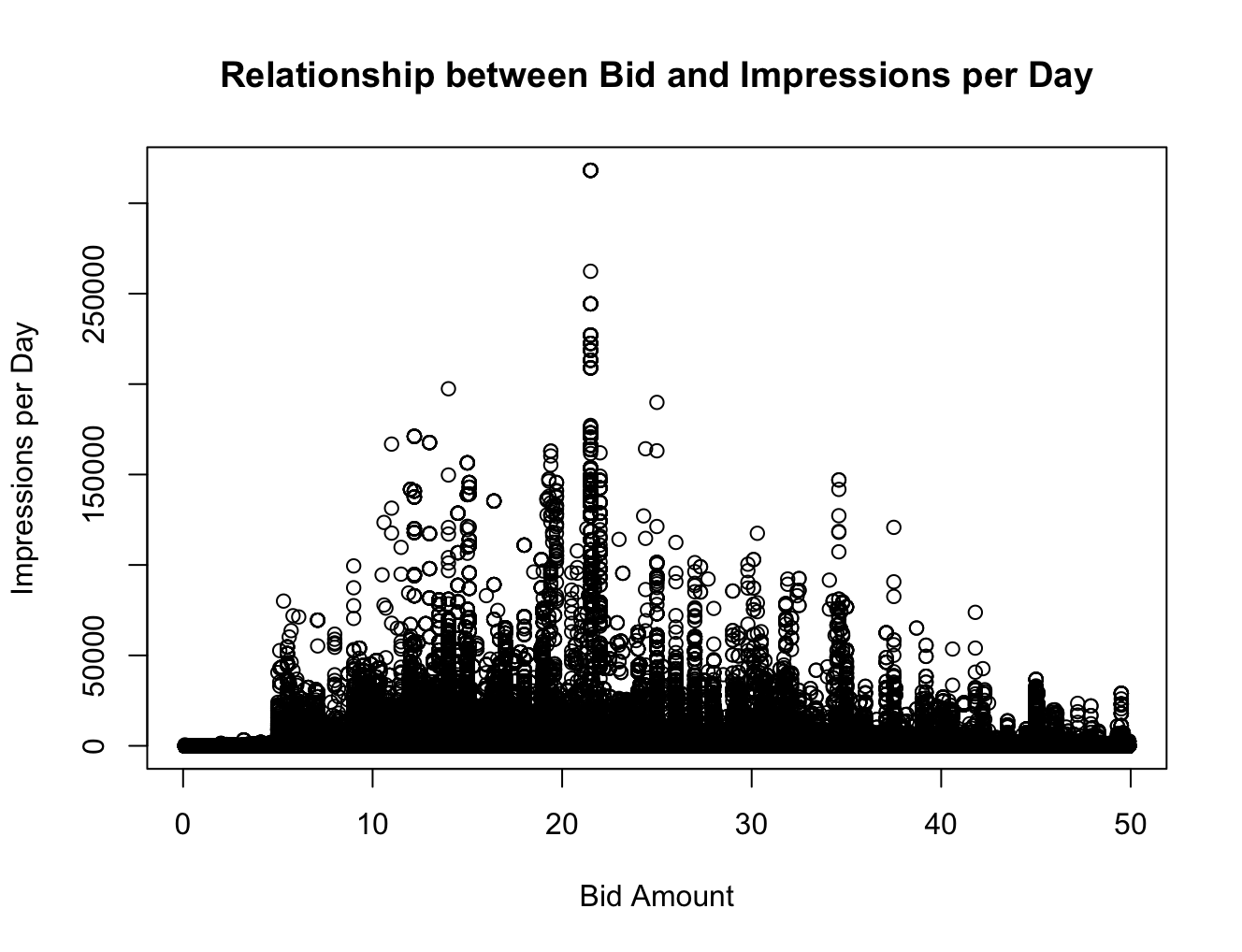}
		\caption{Bid vs. Impressions per Day}
	\end{minipage}
	\hfill
	\begin{minipage}[b]{0.32\textwidth}
		\includegraphics[width=\textwidth]{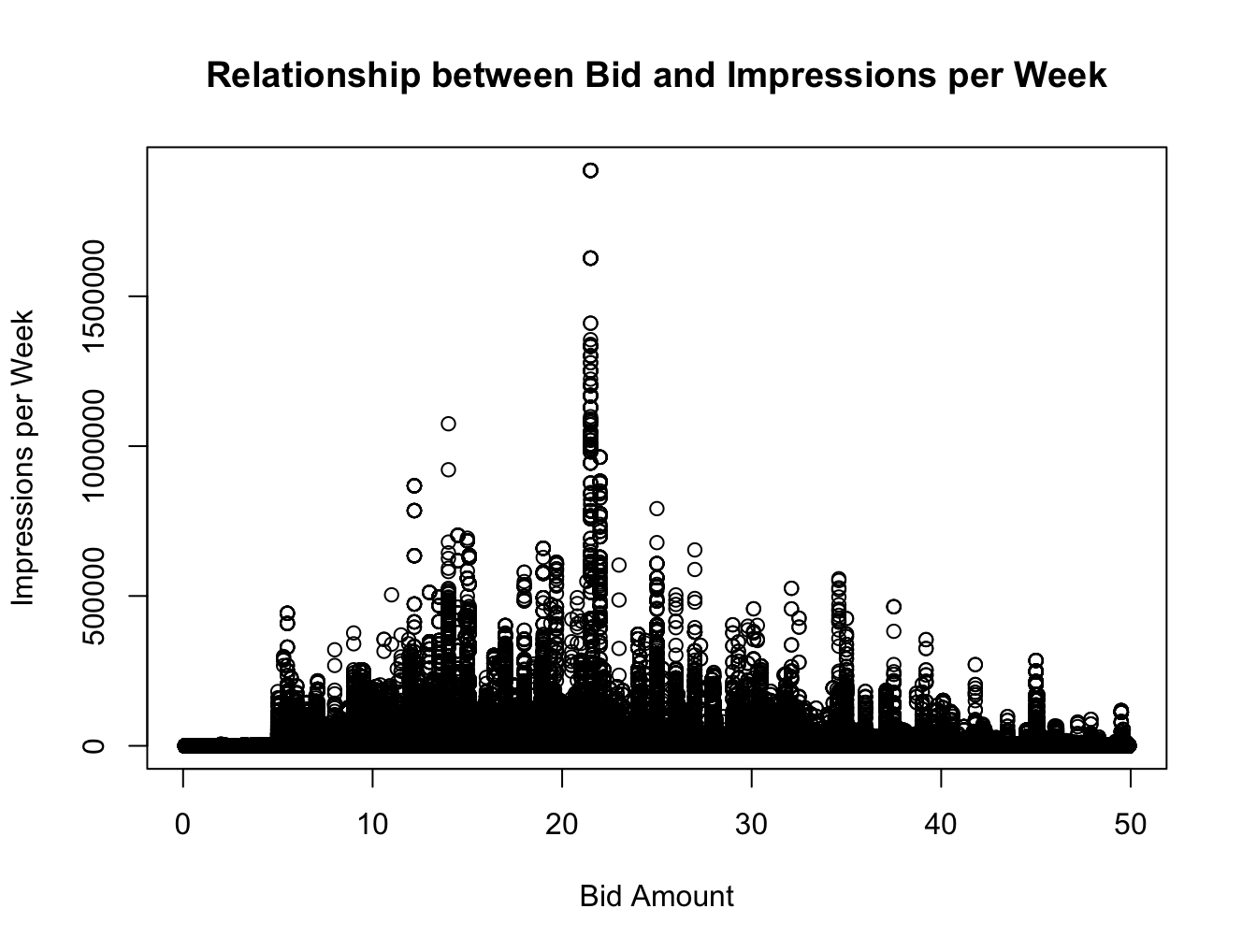}
		\caption{Bid vs. Impressions per Week}
	\end{minipage}
	
	\vspace{1em}
	
	\begin{minipage}[b]{0.32\textwidth}
		\includegraphics[width=\textwidth]{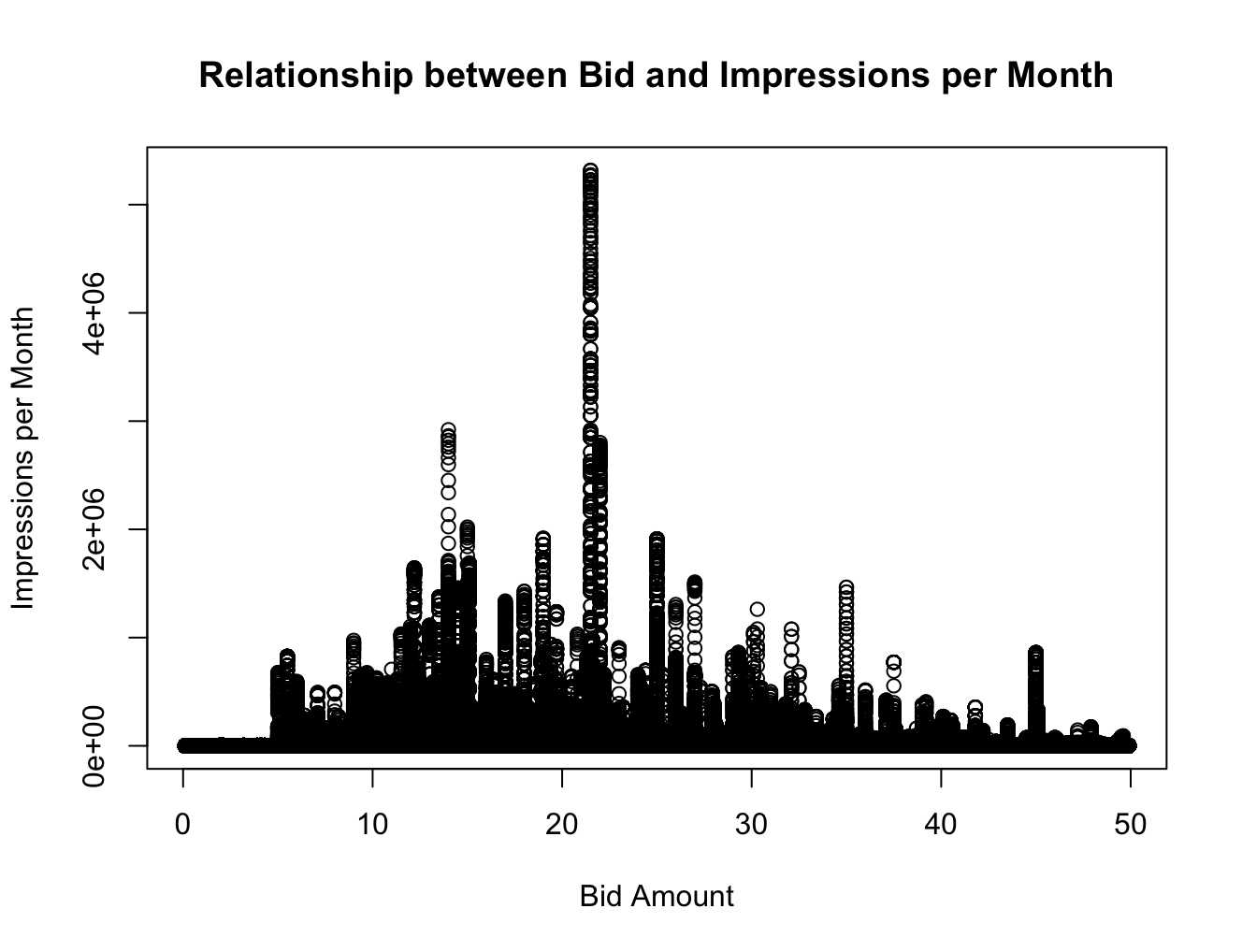}
		\caption{Bid vs. Impressions per Month}
	\end{minipage}
	\hfill
	\begin{minipage}[b]{0.32\textwidth}
		\includegraphics[width=\textwidth]{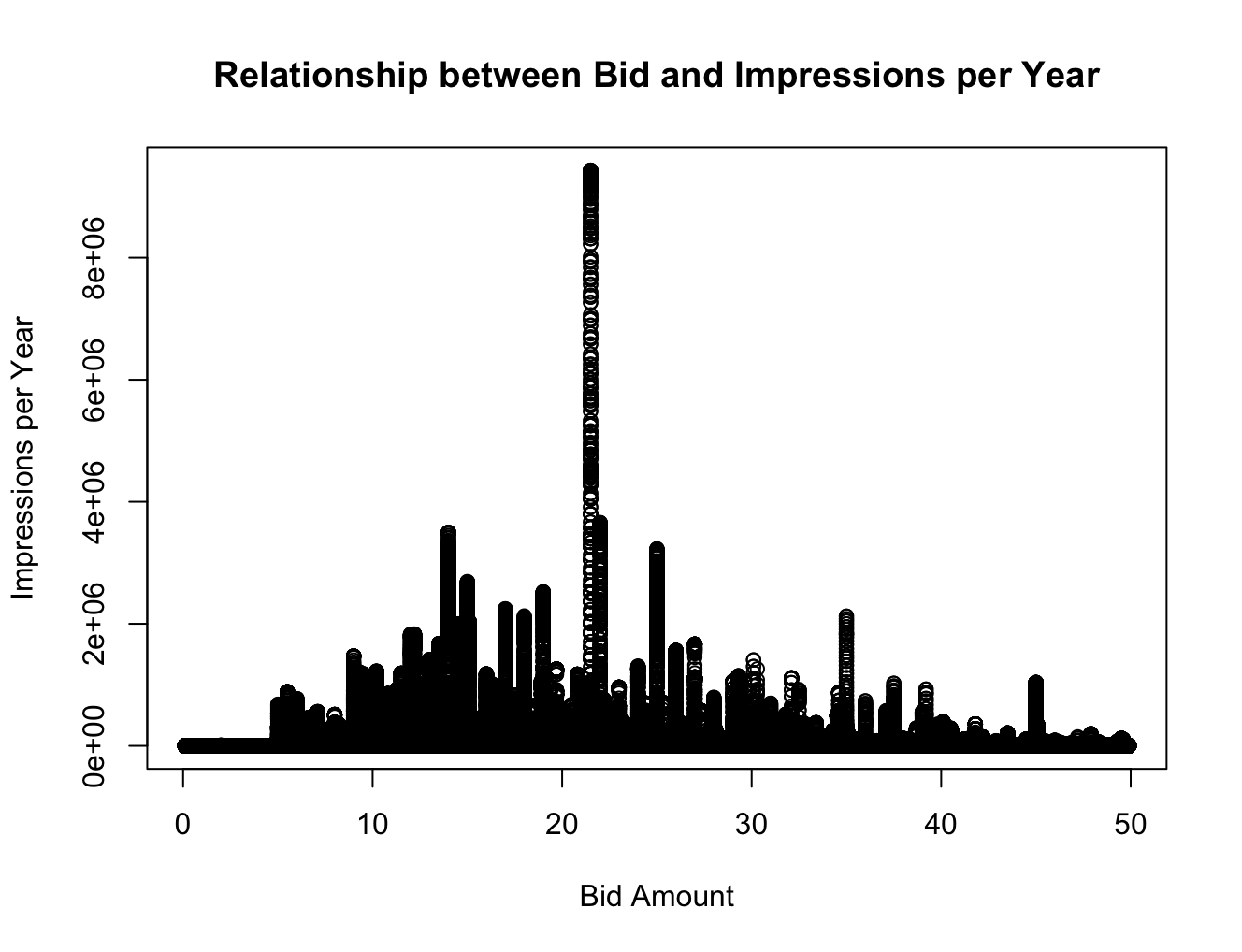}
		\caption{Bid vs. Impressions per Year}
	\end{minipage}
	\hfill
	\begin{minipage}[b]{0.32\textwidth}
		\includegraphics[width=\textwidth]{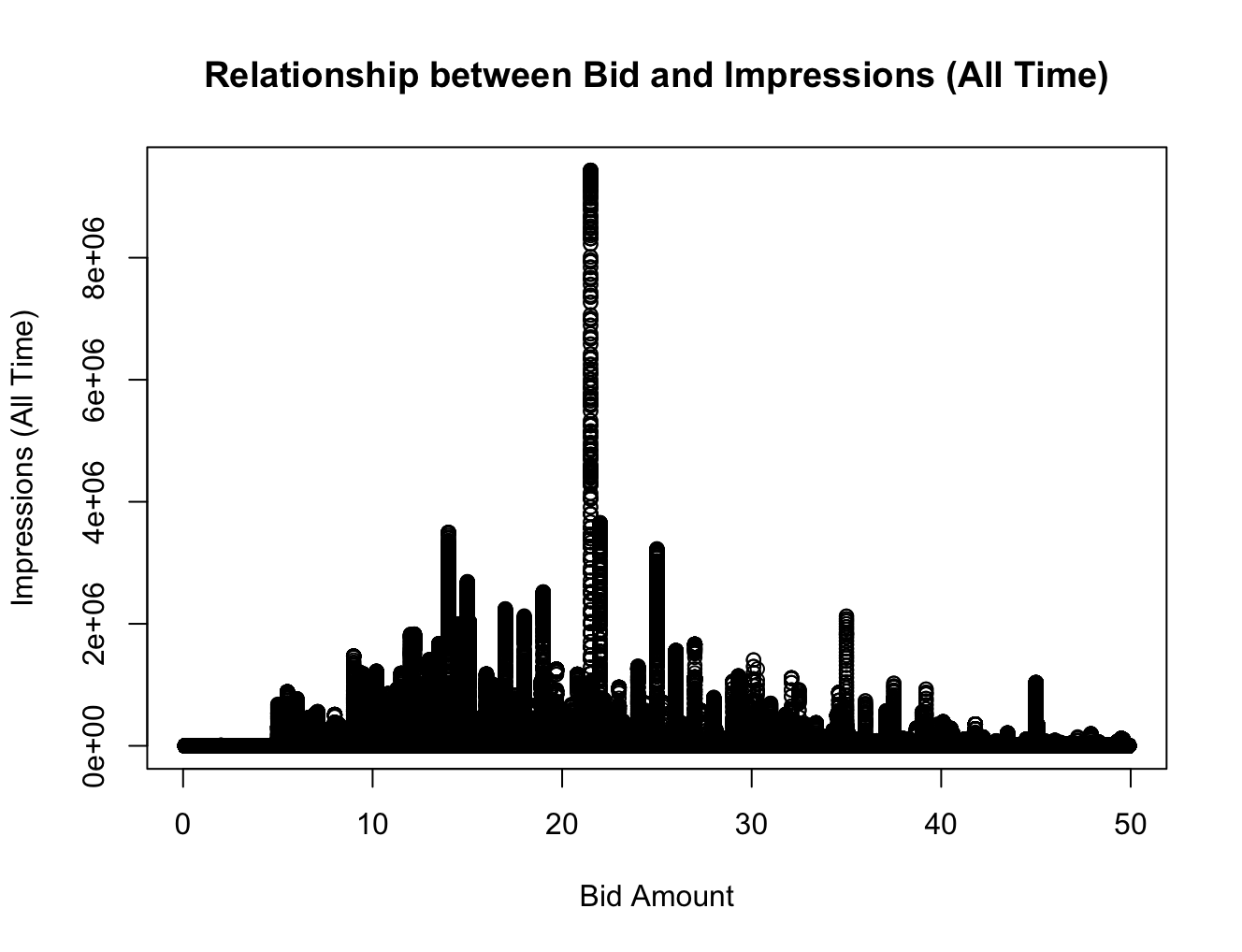}
		\caption{Bid vs. Impressions (All Time)}
	\end{minipage}
\end{figure}

These patterns are best explained through the lens of the Minority Game (MG) framework \cite{challet1997emergence, arthur1994inductive}, a well-established computational model for agent-based adaptation under bounded rationality. The MG originates from the El Farol Bar problem, where agents independently decide whether to participate in a limited-resource environment, achieving success when their actions place them in the minority. Analogously, RTB bidders adaptively optimize their strategies by targeting bid ranges and temporal windows where competitive pressure is minimized.\\
\indent In this context, advertisers are modeled as heterogeneous, inductive agents operating on Demand Side Platforms (DSPs), making repeated decisions based on historical patterns of impressions, auction outcomes, and market density. Success in this environment arises not from following the majority but from deviating intelligently by bidding in sparsely populated regions of the bid landscape to capture impressions at lower cost.\\
\indent Figures~1 through~6 visualize the impression density at different bid levels for a representative bidder. As the time horizon expands (e.g., from hourly to yearly), the stability of the bidder’s strategy becomes increasingly evident. While early-stage patterns (Figures~1–3) exhibit greater noise and volatility, long-term aggregation (Figures~4–6) reveals consistent bid concentration zones, suggesting learned preferences and strategic positioning aligned with minority outcomes.\\
\indent Advertisers base these decisions on high-dimensional data extracted from the bid request issued by the Supply Side Platform (SSP), including user-level features (e.g., age, gender, browsing history), contextual factors (e.g., device type, geolocation, operating system), and auction-specific metadata (e.g., floor price, ad format, position relative to fold). These factors inform audience segmentation, opportunity valuation, and final bid construction. 
In support of the proposed bidding behavior, a recent study shows that the advertiser also does a location-based mobile advertisement.  Location-based advertising on mobile devices has emerged as an essential marketing tool for targeting potential customers. The design of such advertising campaigns is complex, and their effectiveness depends on the ability to collect and examine data that aids in targeting the right customers at the right time and place \cite{keshanian2023mobile}.\\
\indent Given the mutually exclusive nature of ad impressions and the inefficiency of universal high-bid competition, it is not optimal for all bidders to compete uniformly. Instead, bidders partition the bid space $[0, 50]$ (with 0.10 increments) into subranges and selectively target those with lower historical competition. Empirical results suggest that many bidders concentrate within the \$15–\$25 range, while others strategically operate below \$10 or above \$30—indicative of self-organized segmentation and emergent minority group formation.\\
\indent This strategic segmentation functions as a dynamic bid shading mechanism, enabling advertisers to minimize expenditure while maintaining access to desirable impressions. The system evolves into an adaptive, co-adaptive learning environment where agents internalize historical auction feedback to revise strategies—an archetype of an adaptive complex system.\\
\indent In the subsequent section, we simulate a Minority Game model to validate the empirical findings. Agents are assigned distinct bid range preferences and exposed to varying impression landscapes. The simulation demonstrates that advertisers adopting minority strategies outperform majority-aligned bidders in long-run impression acquisition and cost efficiency, substantiating the presence of adaptive minority behavior in RTB auctions.
\subsection{Mathematical Structure}
Consider a scenario in which agents (bidders) must optimize their bids to maximize hourly impressions. The Dynamic Assignment Problem can be formulated as follows:
\begin{itemize}
	\item \textbf{Agents}: A set of bidders \( A = \{a_1, a_2, \ldots, a_n\} \).
	\item \textbf{Tasks}: A set of bidding opportunities \( T = \{t_1, t_2, \ldots, t_m\} \).
	\item \textbf{Time}: A discrete time horizon \( \mathcal{T} \) representing the bidding rounds.
	\item \textbf{Bid Function}:\( c(a_i, t_j, \tau) \) representing the bid amount for agent \( a_i \) to win task \( t_j \) at time \( \tau \).
	\item \textbf{Budget Constraints}: Budget limits and impression targets for each agent.
\end{itemize}
The objective is to determine an assignment \( \pi: A \times \mathcal{T} \to T \) that maximizes the total hourly impressions while minimizing the total bid amount and satisfying the constraints. Let us examine a hypothetical example in the following table.
\begin{table}[H]
	\centering
	\caption{Assignment of bidders to bidding opportunities over time.}
	\label{tab:Table1}
	\begin{tabular}{>{\centering\arraybackslash}m{2cm} >{\centering\arraybackslash}m{2cm} >{\centering\arraybackslash}m{2cm} >{\centering\arraybackslash}m{2cm} >{\centering\arraybackslash}m{2cm} >{\centering\arraybackslash}m{2cm}}
		\toprule
		\textbf{Time Step ($\tau$)} & \textbf{Agent ($a_i$)} & \textbf{Task ($t_j$)} & \textbf{Bid Amount ($c(a_i, t_j, \tau)$)} & \textbf{Hourly Impressions} & \textbf{Assignment Decision} \\
		\midrule
		1 & $a_1$ & $t_1$ & \$5.75 & 120 & Yes \\
		1 & $a_2$ & $t_2$ & \$6.00 & 130 & Yes \\
		1 & $a_3$ & $t_3$ & \$5.70 & 110 & No \\
		\midrule
		2 & $a_1$ & $t_4$ & \$5.50 & 125 & Yes \\
		2 & $a_2$ & $t_5$ & \$5.80 & 135 & Yes \\
		2 & $a_3$ & $t_6$ & \$5.45 & 115 & No \\
		\midrule
		3 & $a_1$ & $t_7$ & \$5.75 & 122 & Yes \\
		3 & $a_2$ & $t_8$ & \$6.20 & 140 & Yes \\
		3 & $a_3$ & $t_9$ & \$5.60 & 118 & No \\
		\midrule
		\vdots & \vdots & \vdots & \vdots & \vdots & \vdots \\
		\midrule
		$T$ & $a_1$ & $t_m$ & \$5.90 & 128 & Yes \\
		$T$ & $a_2$ & $t_{m+1}$ & \$6.30 & 145 & Yes \\
		$T$ & $a_3$ & $t_{m+2}$ & \$5.70 & 120 & No \\
		\bottomrule
	\end{tabular}
	
\end{table}
The interpretation of results from the assignment table necessitates an analysis of bid amounts, hourly impressions, and assignment decisions for each agent across discrete time steps. The following presents a systematic interpretation of the results:
\subsection*{Interpretation of Results}
\paragraph{Time Step Analysis}
\subsubsection*{Time Step 1}
\begin{itemize}
	\item \textbf{Agent \(a_1\)}: Submitted a bid of \$5.75 for task \(t_1\), resulting in 120 hourly impressions. The agent was allocated the task.
	\item \textbf{Agent \(a_2\)}: Submitted a bid of \$6.00 for task \(t_2\), resulting in 130 hourly impressions. The agent was allocated the task.
	\item \textbf{Agent \(a_3\)}: Submitted a bid of \$5.70 for task \(t_3\), resulting in 110 hourly impressions. The agent was not allocated the task.
\end{itemize}
\subsubsection*{Time Step 2}
\begin{itemize}
	\item \textbf{Agent \(a_1\)}: Submitted a bid of \$5.50 for task \(t_4\), resulting in 125 hourly impressions. The agent was allocated the task.
	\item \textbf{Agent \(a_2\)}: Submitted a bid of \$5.80 for task \(t_5\), resulting in 135 hourly impressions. The agent was allocated the task.
	\item \textbf{Agent \(a_3\)}: Submitted a bid of \$5.45 for task \(t_6\), resulting in 115 hourly impressions. The agent was not allocated the task.
\end{itemize}
\subsubsection*{Time Step 3}
\begin{itemize}
	\item \textbf{Agent \(a_1\)}: Submitted a bid of \$5.75 for task \(t_7\), resulting in 122 hourly impressions. The agent was allocated the task.
	\item \textbf{Agent \(a_2\)}: Submitted a bid of \$6.20 for task \(t_8\), resulting in 140 hourly impressions. The agent was allocated the task.
	\item \textbf{Agent \(a_3\)}: Submitted a bid of \$5.60 for task \(t_9\), resulting in 118 hourly impressions. The agent was not allocated the task.
\end{itemize}
\subsubsection*{Time Step \(T\)}
\begin{itemize}
	\item \textbf{Agent \(a_1\)}: Submitted a bid of \$5.90 for task \(t_m\), resulting in 128 hourly impressions. The agent was allocated the task.
	\item \textbf{Agent \(a_2\)}: Submitted a bid of \$6.30 for task \(t_{m+1}\), resulting in 145 hourly impressions. The agent was allocated the task.
	\item \textbf{Agent \(a_3\)}: Submitted a bid of \$5.70 for task \(t_{m+2}\), resulting in 120 hourly impressions. The agent was not allocated the task.
\end{itemize}
\paragraph{Agent Performance}
\begin{itemize}
	\item \textbf{Agent \(a_1\)}: Consistently submits bids in the range of \$5.50 to \$5.90 and is allocated tasks in all time steps. The hourly impressions range from 120 to 128, indicating consistent performance.
	\item \textbf{Agent \(a_2\)}: Submits higher bids ranging from \$5.80 to \$6.30 and is allocated tasks in all time steps. The hourly impressions range from 130 to 145, demonstrating superior performance.
	\item \textbf{Agent \(a_3\)}: Submits bids within the range of agent \(a_1\) but is not allocated tasks in any time step. The hourly impressions are lower, ranging from 110 to 120, suggesting a less effective bidding strategy.
\end{itemize}
\paragraph{Overall Trends}
\begin{itemize}
	\item \textbf{Bid Amounts}: Higher bid amounts generally correlate with increased hourly impressions and a higher probability of task allocation.
	\item \textbf{Hourly Impressions}: Agents submitting higher bid amounts tend to receive more hourly impressions, indicating a positive correlation between bid amounts and impressions.
	\item \textbf{Assignment Decisions}: Agents \(a_1\) and \(a_2\) are consistently allocated tasks, while agent \(a_3\) is not. This suggests that \(a_3\)'s bidding strategy may require modification to enhance competitiveness.
\end{itemize}
The results indicate that agents \(a_1\) and \(a_2\) possess effective bidding strategies that result in consistent task assignments and higher hourly impressions. Agent \(a_3\), however, demonstrates difficulty in securing task assignments, potentially due to a less competitive bidding strategy. To enhance \(a_3\)'s performance, it may be necessary to identify the minority bid range where the bidder would be in the minority and face no competition.\\
Finally, the equilibrium assignment would occur where each agent would be in the minority bid range and the allocation will be non-degeneracy. The equilibrium assignment matrix has been represented using the following matrix for $3\times 3$ and extended for $A\times T$.
\begin{table}[h!]
	\centering
	\caption{Agent Task Allocation Based on Bid Range}
	\begin{tabular}{|l|c|c|c|}
		\hline
		\textbf{Agent $\backslash$ Bid Range} & \textbf{Task $t_1 = [5.50-5.90]$} & \textbf{Task $t_2 = [5.80-6.30]$} & \textbf{Task $t_3$} \\ 
		\hline
		$a_{1}$ & Yes & No & No \\ 
		\hline
		$a_{2}$ & No & Yes & No \\ 
		\hline
		$a_{3}$ & No & No & Yes \\ 
		\hline
	\end{tabular}
	\label{tab:agent_task_allocation}
\end{table}
The Table \ref{tab:agent_task_allocation} has been derived using the theoretical understanding about the dynamic assignment problem from Table 1. It suggest that if the agent $a_{3}$  should find any other bid range for wining ad impression.  In the following subsection  this equilibrium assignment/choice has been derived using Minority Game Theory by applying Algorithm 1 and the Algorithm 2. Finally it has been shown that if Algorithm 1 and 2 are true then the equilibrium assignment will be exactly like this sample three agents assignment problem. This will prove theoretically that the bidders in the ad exchange market is taking minority strategy and identifies their respective bid range to be minor. This is why the bid behavior has been found as shown from Figure 1 to 6. Finally empirically this theoretical model has been proved in section 5.
\subsection{Minority Game and Bidding Strategy}
\paragraph{Minority Game Structure}
Let N agents be in the game, labeled by an index $i = 1,\cdots, N$. This number $N$, usually considered odd, is assumed to be large. This $N$ includes agents who are buying and selling the ad impressions. The game proceeds at discrete iteration steps, denoted by $\tau = 0, 2,.\cdots$. We now translate the main structure of the problem into simple mathematical rules. These ingredients are:
\begin{itemize}
	\item \textit{Agent's actions:} At each step $\tau$ each agent takes a binary decision $b_{i}(\tau) \in \{-1, 1\}$. This is called buying $(-1)$ versus selling $(+1)$decisions. An aggregate of all the decisions at step $\tau$ is defined as $A(\tau) = \dfrac{\sum_{i=1}^{N}b_{i}(\tau)}{\sqrt{N}}$. $A(\tau)$ is divided by $\sqrt{N}$ instead of $N$ because if $N\rightarrow \infty$ then $A(\tau)$ would not lead to zero. This is primarily rooted in statistical and scaling considerations.
	\item \textit{Public Information:} Agents base their decisions on historical market information. Here, the market information includes the past bid and ad impression offers. Moreover, it includes the conversion rate as well in terms of CPC, CPA, etc They are given at step $\tau$ the signs of the overall bids (sale or buy decisions, i.e., whether buyers or sellers were the minority group) over the $M$ most recent steps. This means, \begin{equation*}
		(sign[A(\tau-M)],\cdots,sign[A(\tau-1)]) \in \{-1, 1\}^{M}
	\end{equation*}
	\item \textit{Agent's strategies:} Each agent $i$ has $S$ strategies $\Delta^{ia}$, with $a = 1,\cdots, S$. A strategy $\Delta^{ia}$ defines a mapping from information strings to a recommended trading action:
	\begin{equation*}
		\Delta^{ia}: \{-1, 1\}^{M} \rightarrow \{-1,1\}
	\end{equation*}
	A strategy is a lookup table with $2^{M}$ entries, each being $\pm$. These entries remain fixed throughout the game. If agent $i$ uses strategy $a$ at step $\tau$, then he will act deterministically according to 
	\begin{equation*}
		b_{i}(\tau)= \Delta^{ia}(sign[A(\tau-M)],\cdots,sign[A(\tau-1)]) 
	\end{equation*}
\end{itemize}
\paragraph{Algorithm 1} 
The \textit{Minority Bidding Simulation Algorithm} is designed to simulate a bidding process in which multiple agents compete to win by bidding below the median bid. This algorithm is particularly useful for examining the dynamics of competitive bidding environments and the strategies agents may employ to maximize their success in achieving a low bid while securing high impressions.\\	
Agents adaptively adjust their bids based on recent performance history. If an agent fails to win in the last \texttt{HISTORY\_LENGTH} rounds, it modifies its bid randomly within a bounded interval. This perturbation enables strategic exploration of the bid space and encourages convergence toward effective bidding behavior.\\
The algorithm concurrently tracks round-wise metrics, including the average bid and the number of winners. These indicators provide insight into the emerging dynamics of bid distribution and strategic diversity among agents.\\
Using the specified parameter configuration, Algorithm 2 is executed to simulate the evolution of bidding behavior over repeated rounds.
\paragraph*{Parameters:}
\begin{itemize}
	\item NUM\_AGENTS: The number of bidding agents, set to 100.
	\item NUM\_ROUNDS: The number of bidding rounds, set to 50.
	\item HISTORY\_LENGTH: The memory length for decision-making, set to 5.
\end{itemize}
\begin{algorithm}[H]
	\caption{Minority Bidding Simulation Algorithm}
	\begin{algorithmic}[1]
		\State Initialize agents with random bids and empty memory
		\For{$round = 1$ to $NUM\_ROUNDS$}
		\State \textsc{RunAuctionRound}(agents)
		\State Record average bid and number of winners
		\State Print round results
		\EndFor
		
		\Procedure{RunAuctionRound}{agents}
		\State Collect bids from all agents
		\State Calculate median bid
		\State Determine winners (agents with bids below the median)
		\For{each agent}
		\State Update success status
		\State Update memory with recent outcome
		\State Increment success count if winner
		\If{recent performance is zero}
		\State Adjust bid randomly within $[-0.5, 0.5]$
		\State Clamp bid to range $[5, 10]$
		\EndIf
		\EndFor
		\EndProcedure
	\end{algorithmic}
\end{algorithm}
\section*{Minority Bidding Simulation Algorithm}	
\begin{algorithm}[H]
	\caption{Minority Bidding Simulation}
	\label{algo:minority_bidding}
	\scriptsize 
	\resizebox{\textwidth}{!}{
		\begin{minipage}{\textwidth}
			\begin{algorithmic}[1]
				\State \textbf{Parameters:}
				\State $NUM\_AGENTS \gets 100$ \Comment{Number of bidding agents}
				\State $NUM\_ROUNDS \gets 50$ \Comment{Number of bidding rounds}
				\State $HISTORY\_LENGTH \gets 5$ \Comment{Memory length for decision-making}
				
				\State \textbf{Initialize agents:}
				\For{$i = 1$ to $NUM\_AGENTS$}
				\State Initialize agent $i$ with:
				\State \hspace{1cm} $bid \gets \text{random value between [5, 10]}$
				\State \hspace{1cm} $memory \gets []$ \Comment{Tracks recent outcomes}
				\State \hspace{1cm} $success\_count \gets 0$ \Comment{Number of successful bids}
				\EndFor
				
				\Procedure{RunAuctionRound}{agents}
				\State \textbf{Collect bids:}
				\State $bids \gets [bid_1, bid_2, \dots, bid_{NUM\_AGENTS}]$ 
				\State $median\_bid \gets \text{Median}(bids)$
				
				\State \textbf{Determine winners:}
				\State $winners \gets \{i : bid_i < median\_bid\}$ 
				
				\State \textbf{Update agents:}
				\For{$i = 1$ to $NUM\_AGENTS$}
				\State $success \gets 1 \text{ if } i \in winners \text{ else } 0$
				\State Append $success$ to $agent[i].memory$
				\State $agent[i].success\_count \gets agent[i].success\_count + success$
				
				\If{Length of $agent[i].memory \geq HISTORY\_LENGTH$}
				\State $recent\_performance \gets \text{Sum of last HISTORY\_LENGTH outcomes}$
				\If{$recent\_performance == 0$}
				\State Adjust $agent[i].bid$ by $random\_adjustment \in [-0.5, 0.5]$
				\State Clamp $agent[i].bid$ to range $[5, 10]$
				\EndIf
				\EndIf
				\EndFor
				\EndProcedure
				
				\State \textbf{Run Simulation:}
				\State $avg\_bids \gets []$, $total\_winners \gets []$
				\For{$round = 1$ to $NUM\_ROUNDS$}
				\State $(avg\_bid, num\_winners) \gets \textsc{RunAuctionRound}(agents)$
				\State Append $avg\_bid$ to $avg\_bids$
				\State Append $num\_winners$ to $total\_winners$
				\State Print ``Round $round$: Avg Bid = $avg\_bid$, Winners = $num\_winners$''
				\EndFor
			\end{algorithmic}
		\end{minipage}
	} 
\end{algorithm}

\begin{figure}[h!]
	\centering
	\begin{minipage}{0.48\linewidth}
		\centering
		\includegraphics[width=\linewidth]{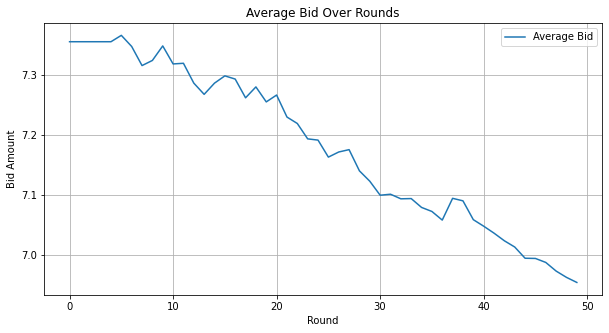} 
		\caption{Average Bids Over Rounds}
		\label{fig:figure9}
	\end{minipage}%
	\hfill
	\begin{minipage}{0.48\linewidth}
		\centering
		\includegraphics[width=\linewidth]{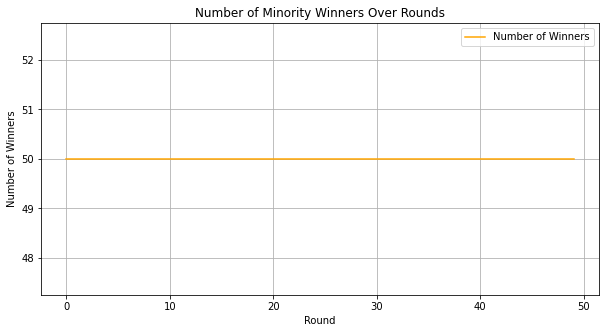} 
		\caption{Number of Minority Winners}
		\label{fig:figure10}
	\end{minipage}
	\caption{Simulation Results}
	\label{fig:figure9_10}
\end{figure}
\begin{itemize}
	\item Given the definitions above, $\Delta^{ia} $ is good at stage $\tau$ if $\Delta^{ia}(sign[A(\tau-M)],\cdots,sign[A(\tau-1)]) = -sign[A(\tau)]$. 
	\item \textit{Strategy valuations} Each agent $i$ keeps track of a valuation $v_{ia}$ for each of his strategies, measuring their track records :
	\begin{equation*}
		v_{ia}(\tau+1)=v_{ia}(\tau)-A(\tau)\Delta^{ia}(sign[A(\tau-M)],\cdots,sign[A(\tau-1)]) 
	\end{equation*}
	\item \textit{Dynamic strategy selection:} At each step $\tau$ of the game each agent $i$ will select his best strategy $a^{i}(\tau)$ at that stage of the process, defined as
	\begin{equation*}
		a_{i}(\tau) = \arg\max_{a \in \{1, \ldots, S\}} v_{ia}(\tau)
	\end{equation*}
\end{itemize}
\indent Based on this game structure, a simulation has been run to test the minority strategy as provided in Algorithm 2 and produces results in Figure \ref{fig:figure9} and \ref{fig:figure10}. These two figures have been put together as "Simulation Results" in Figure \ref{fig:figure9_10}.\\
Modeling an Agent-Based Minority Game with Dynamic Bid Adjustments by Participants. Each iteration involves agents submitting bids for advertising space. The minority group, comprising agents with the lowest aggregate bids, succeeds in the auction. Graphical representations illustrate the progression of mean bids and successful agent quantities over time.\\
The correlation between bids and advertisement impressions for individual agents has been graphically represented. The data visualizations indicate that higher or lower bids do not consistently correspond to increased or decreased impressions. Agents submitting lower bids occasionally receive impressions when they remain in the minority. Certain agents with higher bids fail to secure hourly advertisement placements.\\
From Figures \ref{fig:figure9} and \ref{fig:figure10}, it is clear that as the bidders take a minority strategy, the average bid is reducing over rounds. Initially, when the bidding started, the competition was very high; as a result, the average bid was high. In the subsequent rounds, the bidders can identify the minority strategy and not participate in the majority group, so the average bid value is reduced. From Figure  \ref{fig:figure10}, it is clear that the number of winners is fixed as they have already taken the minority strategy and are there.\\
\paragraph{Algorithm3} The algorithm "Plot Bid vs Hourly Impressions for Each Agent" is designed to visualize the relationship between the bid amounts and hourly impressions for multiple agents. This visualization facilitates understanding how different bidding strategies affect the impressions each agent receives over time. The algorithm inputs a data dictionary and produces scatter plots for each agent, illustrating the bid amount versus hourly impressions.\\
The algorithm comprises several key steps: initialization, setting up the plot grid, extracting and plotting data for each agent, handling unused subplots, and finalizing the plot. Below is a detailed explanation of each part.\\
The scatter plots visually represent how each agent's bid amount correlates with the hourly impressions they receive. This can assist in identifying patterns and optimizing bidding strategies for improved performance.\\
The algorithm focuses on visualizing the relationship between bid amounts and hourly impressions. The scatter plots serve as a performance metric, enabling users to compare different agents and understand the effectiveness of their bidding strategies.\\
\indent If all agents are following the minority strategy, then each agent or bidder will be operating in a given range of bid and impression combinations. Based on Algorithm 3 and  Figure \ref{fig:figure15}, each sample of ten agents has different areas of bid-impression combinations. The area of bidding is different for each of the ten bidders. Thus far, it is evident that if a bidder selects a minority strategy in the ad exchange market, that bidder does not prefer to offer the high bid to gain ad impressions. Based on this, the following section and subsections explain empirical evidence for Algorithms 2 and 3.
\paragraph{Pseudocode: Plot Bid vs Hourly Impressions for Each Agent}
\begin{algorithm}[H]
	\caption{Plot Bid vs Hourly Impressions}
	\label{fig:figure14}
	\scriptsize 
	\resizebox{\textwidth}{!}{
		\begin{minipage}{\textwidth}
			\begin{algorithmic}[1]
				\State \textbf{Input:}
				\State Data dictionary containing:
				\State \hspace{0.5cm} Columns: \texttt{date}, \texttt{hour}, \texttt{adid\_anonymized}, \texttt{bid}, \texttt{imps\_alltime}, \texttt{imps\_hour}, \texttt{imps\_day}, \texttt{imps\_week}, \texttt{imps\_month}, \texttt{imps\_year}
				\State \textbf{Output:} Scatter plots for each agent showing Bid Amount vs Hourly Impressions
				
				\State \textbf{Initialize:}
				\State $df \gets \text{Convert data dictionary to a DataFrame}$
				\State $num\_agents \gets \text{Length of DataFrame (number of agents)}$
				
				\State \textbf{Setup Plot Grid:}
				\State Create a grid of subplots with $nrows = \frac{num\_agents}{2}$ and $ncols = 2$
				\State Flatten the grid to iterate over all axes
				
				\For{$i = 0$ to $num\_agents - 1$}
				\State \textbf{Extract Agent Data:}
				\State $agent\_bid \gets df['bid'][i]$
				\State $agent\_imps\_hour \gets df['imps\_hour'][i]$
				
				\State \textbf{Plot Data:}
				\State Scatter plot $(agent\_bid, agent\_imps\_hour)$ on axis $axes[i]$
				\State Set title: \texttt{Agent i+1: Bid vs Hourly Impressions}
				\State Set x-label: \texttt{Bid Amount}
				\State Set y-label: \texttt{Hourly Impressions}
				\State Enable grid on axis
				\EndFor
				
				\State \textbf{Handle Unused Subplots:}
				\For{$j = num\_agents$ to $len(axes) - 1$}
				\State Delete unused axis $axes[j]$
				\EndFor
				
				\State \textbf{Finalize Plot:}
				\State Adjust layout for clarity
				\State Show the final figure with all subplots
			\end{algorithmic}
		\end{minipage}
	}
\end{algorithm}
\begin{figure} [H]
	\centering
	\includegraphics[width=0.75 \textwidth]{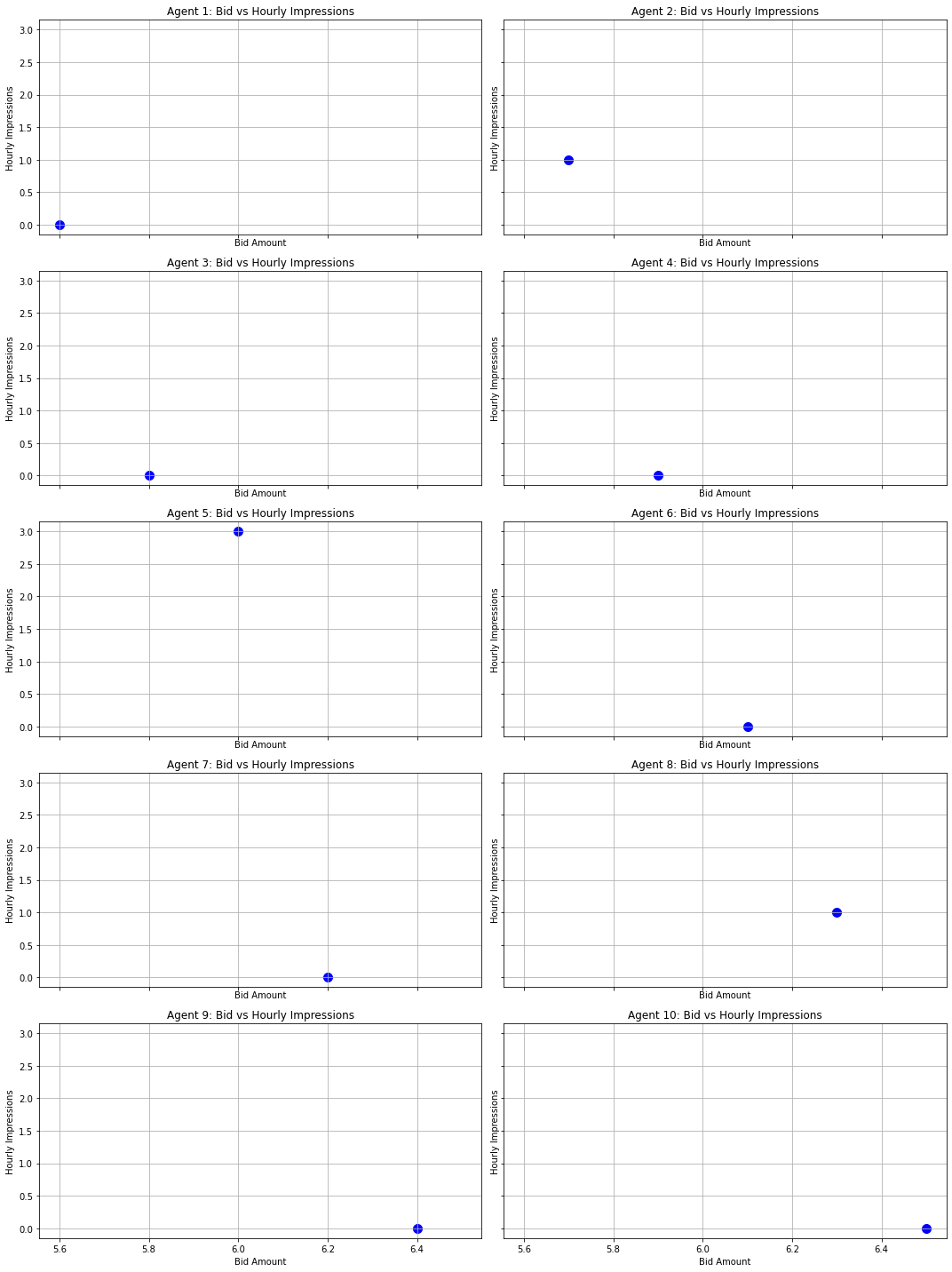}
	\caption{Minority Behavior Detection}
	\label{fig:figure15}
\end{figure} 
\section{Main Results}
The section formalizes the above analytical findings in terms of general theory. 
\begin{proposition}[Existence of Minority Partition]
	Let $A = \{a_1,\ldots,a_N\}$ denote the set of bidders in a repeated first-price RTB auction, and let $B=[0,50] \text {or} B=[0,\bar{B}]$ be the discrete bid space with increments $\delta=0.10$. 
	If each bidder chooses a bid range according to the Minority Game dynamics with finite memory $M$ and finite strategies $S$, then the bid space admits a partition. 
	\[
	B = \bigcup_{k=1}^{K} B_k, \quad B_i \cap B_j = \varnothing, \; i\neq j,
	\]
	In every round, a non-empty subset of agents occupies a minority sub-range $B_k$. Here $\quad B_i \cap B_j = \varnothing, \; i\neq j,$ signifies that each bid range is unique for some minor bidder. In addition, each minor bidder selects only a unique bid range and does not participate in all the minor bid ranges.
\end{proposition}

\begin{proof}
The proof is given step by step below.
\end{proof}
\subsection*{Existence of a Minority Partition (Discrete Bins Case)}

Let $A$ be a finite set of agents and $K$ denote a finite set of bid bins. 
(e.g., a discrete bid space). 
A \emph{bin state} is a function.
\[
s: A \to K,
\]
This assigns each agent the bin in which her bid falls. This choice is independent and not exogenously determined. This means each bidder selects the bid range or bin according to their choice.

\begin{definition}[Bin occupancy.]
For each bin $k \in K$, define the \emph{Bin Count} as
\[
\mathrm{BinCount}(s,k) = \bigl|\{\, i \in A : s(i) = k \,\}\bigr|,
\]
That is, the number of agents mapped into bin $k$ under the assignment $s$. This means the assignment is non-homogeneous. Each $i$ selects a particular bin, assuming their parameters are given. So it is choice-dependent.
\end{definition}

\begin{definition}[Minority bin.]

	A bin $k \in K$ is called an \emph{occupied minority bin} if
	\[
	\mathrm{BinCount}(s,k) > 0 
	\quad\text{and}\quad
	\mathrm{BinCount}(s,k) \leq \mathrm{BinCount}(s,k')
	\quad\text{for all } k' \in K \text{ with } \mathrm{BinCount}(s,k')>0.
	\]
	In other words, $k$ is occupied by at least one agent, and no other occupied bin has strictly fewer agents.
\end{definition}

\begin{lemma}[Existence of occupied minority bin]
	For every finite set of agents $A$, every finite set of bins $K$, 
	and every assignment $s : A \to K$, 
	There exists at least one occupied minority bin.
\end{lemma}

\begin{proof}
	Since $A$ is nonempty, pick any agent $i_0 \in A$ and let $k_0 = s(i_0)$. 
	Then $\mathrm{BinCount}(s,k_0) \ge 1$, so the set of occupied bins 
	\[
	\{\, k \in K : \mathrm{BinCount}(s,k) > 0 \,\}
	\]
	Is nonempty. Because $K$ is finite, this set admits a bin $k_{\min}$ minimizing the count. 
	By construction $\mathrm{BinCount}(s,k_{\min})>0$ 
	and for every occupied $k'$, one has 
	$\mathrm{BinCount}(s,k_{\min}) \le \mathrm{BinCount}(s,k')$. 
	Thus $k_{\min}$ is an occupied minority bin.\\
	$BinCount(s,k_{\min})$ is feasible because the set $B$ is closed and  bounded, i.e. compact.
\end{proof}

\begin{proposition}[Existence of Minority Partition, discrete bins]
	In any round of the repeated auction, represented by an assignment 
	$s: A \to K$ of agents into finitely many bins, 
	There exists at least one occupied minority bin. 
	In every round, there is a nonempty sub-range of the discretized bid space 
	occupied by strictly fewer bidders than all others.
\end{proposition}
\begin{remark}
	This means Propositions 1 \& 2, along with Lemma 1, confirm that in each round of the auction, there exists a minority of bidders who try to be in the minority bid range. This is acting as the opposite of the standard truthful bidding mechanism. Many bidders do not reveal their preferences in that bid range. As a result, the bid range is acting as a minority bid range.  
\end{remark}
\begin{lemma}[Endogenous Bid Shading]
	Let $v_i$ denote the private valuation of bidder $a_i$ and $b_i(\tau)$ the observed bid at round $\tau$. 
	Under the equilibrium partition identified in Proposition 1, the bid satisfies.
	\[
	\mathbb{E}[b_i(\tau)] < v_i, \quad \forall i \in A,
	\]
	With strict inequality whenever $a_i$ persists in a minority sub-range. This means selecting a minority bid range indicates the strategy of bid shading.
\end{lemma}

\begin{proof}
Let $A$ denote the finite set of bidders. For each agent $i \in A$, let $v_i$ denote
her private valuation, and let $b_i(\tau)$ denote her bid in round $\tau$.
Consider an expectation operator $\mathbb{E}$, which maps the stochastic bid
process $\{b_i(\tau)\}_{\tau \geq 0}$ to its expected value for each agent.
Finally, let $\mathrm{Pers}(i)$ be a predicate indicating whether agent $i$
persists in a minority sub-range of the equilibrium partition.

We assume two standard hypotheses, consistent with the equilibrium identified
in Proposition~1:

\begin{enumerate}
	\item \textbf{Global shading.} In equilibrium, each bidder's expected bid is
	Below her valuation:
	\[
	\mathbb{E}[b_i(\tau)] < v_i \quad \text{for all } i \in A.
	\]
	
	\item \textbf{Minority strictness.} If bidder $i$ persists in a minority sub-range,
	Then the inequality is strict:
	\[
	\mathrm{Pers}(i) \;\;\Rightarrow\;\; \mathbb{E}[b_i(\tau)] < v_i.
	\]
\end{enumerate}

Under these assumptions, we obtain that for every agent $i \in A$,
\[
\mathbb{E}[b_i(\tau)] < v_i,
\]
and in particular the inequality is strict whenever $\mathrm{Pers}(i)$ holds. It is true because competition in $BinCount(s,k_{\min})$  is lower than any other $BinCount(s,k^{'})$. Hence, from the existing theory, truthful bidding is directly related to the number of bidders participating in an auction.
\end{proof}
\begin{remark}
	\begin{itemize}
		\item If the equilibrium delivers only weak shading globally,
		$\mathbb{E}[b_i(\tau)] \leq v_i$ for all $i$, but strict shading on the
		minority sub-range, the conclusion adapts to: weak inequality globally,
		strict inequality when $\mathrm{Pers}(i)$ holds.
		\item If the expectation operator $\mathbb{E}$ is instantiated as a true
		probabilistic expectation with respect to a distribution over bids, the
		The lemma remains valid; one supplies the two assumptions as
		consequences of the equilibrium analysis.
		\item The lemma is agnostic about \emph{why} shading occurs: the economic
		Reasoning is encapsulated in the two assumptions. This modular structure
		allows the lemma to be applied once the global shading and minority
		Strictness properties are verified for a given model.
	\end{itemize}
\end{remark}

\begin{theorem}[Minority Strategy Efficiency]
	In a repeated RTB auction with $N$ bidders, suppose a subset $\mathcal{M} \subseteq A$ adopts minority strategies while $\mathcal{J}=A \setminus \mathcal{M}$ adopts majority-aligned strategies. 
	Let $\mathrm{Imps}(a,t)$ denote impressions secured by agent $a$ at round $t$, and define the per-round efficiency
	\[
	e_a(t) \;:=\; \frac{\mathrm{Imps}(a,t)}{b_a(t)}, \qquad b_a(t) > 0.
	\]
	Then in the long run,
	\[
	\lim_{\tau \to \infty} \frac{1}{\tau}\sum_{t=1}^{\tau} e_{a_i}(t)
	\;>\;
	\lim_{\tau \to \infty} \frac{1}{\tau}\sum_{t=1}^{\tau} e_{a_j}(t),
	\]
	for all $a_i \in \mathcal{M}, \; a_j \in \mathcal{J}$.
	That is, minority strategies yield strictly higher impression efficiency per unit cost.
\end{theorem}

\begin{proof}[The proof is based on two sets of assumptions]
	
	\paragraph{Route A: Deterministic eventual gap.}
	Assume there exist $\varepsilon>0$ and $T_0\in\mathbb N$ such that for all $t\ge T_0$,
	\[
	e_{a_i}(t) \;\ge\; e_{a_j}(t) + \varepsilon, 
	\quad\text{for all } a_i \in \mathcal M,\; a_j \in \mathcal J.
	\]
	For $\tau\ge T_0$,
	\[
	\frac1{\tau}\sum_{t=1}^{\tau} \bigl(e_{a_i}(t)-e_{a_j}(t)\bigr)
	=\frac{C}{\tau}+\frac{\tau-T_0+1}{\tau}\,\varepsilon,
	\]
	where $C=\sum_{t=1}^{T_0-1}(e_{a_i}(t)-e_{a_j}(t))$ is constant.
	Letting $\tau\to\infty$ yields
	\[
	\liminf_{\tau\to\infty}\frac1{\tau}\sum_{t=1}^{\tau} \bigl(e_{a_i}(t)-e_{a_j}(t)\bigr)\;\ge\;\varepsilon.
	\]
	If the Cesàro means converge, this $\liminf$ equals the difference of the limits, so
	\[
	\lim_{\tau\to\infty}\frac1{\tau}\sum_{t=1}^{\tau} e_{a_i}(t)
	\;\ge\;
	\lim_{\tau\to\infty}\frac1{\tau}\sum_{t=1}^{\tau} e_{a_j}(t) + \varepsilon,
	\]
	Establishing strict inequality.
	
	\paragraph{Route B: Stationary–ergodic law of large numbers.}
	Assume instead that each $\{e_a(t)\}$ is strictly stationary and ergodic with $\mathbb E[|e_a(1)|]<\infty$, and that
	\[
	\mathbb E[e_{a_i}(1)] \;\ge\; \mathbb E[e_{a_j}(1)] + \varepsilon
	\quad\text{for all } a_i \in \mathcal M,\; a_j \in \mathcal J
	\]
	for some $\varepsilon>0$. By Birkhoff’s ergodic theorem,
	\[
	\frac1{\tau}\sum_{t=1}^{\tau} e_a(t) \;\xrightarrow{\text{a.s.}}\; \mathbb E[e_a(1)].
	\]
	Thus
	\[
	\frac1{\tau}\sum_{t=1}^{\tau} \bigl(e_{a_i}(t)-e_{a_j}(t)\bigr) \;\xrightarrow{\text{a.s.}}\; \mathbb E[e_{a_i}(1)]-\mathbb E[e_{a_j}(1)] \;\ge\; \varepsilon,
	\]
	This yields the desired strict inequality.
\end{proof}
\begin{corollary}[Adaptive Stability of Minority Strategies]
	The equilibrium established in Theorem 1 is evolutionary stable: if a majority bidder deviates into a minority range, they improve efficiency until the minority becomes overcrowded, restoring the equilibrium partition.
\end{corollary}

\begin{proof}
[Adaptive Stability of Minority Strategies]
Assuming that the bidders are in continuum, i.e. $a\in[0,1]$ denotes the fraction of bidders in the minority set $\mathcal M$.
Define the efficiency gap.
\[
\Delta(a)_{e} \;:=\; \mathbb E[e_\mathcal{M}(a)] - \mathbb E[e_J(a)],
\]
Where $e_\mathcal{M}(a)$ (resp.\ $e_J(a)$) is the per-round efficiency for a representative minority (resp.\ majority) bidder at minority share $a$.
Suppose $\Delta$ is continuous, strictly decreasing, and admits a unique zero at $a^\star\in(0,1)$.

Then $a^\star$ is evolutionary stable: if $a<a^\star$, then $\Delta(a)_{e}>0$ and minority strategies are strictly more efficient, so the minority share increases. If $a>a^\star$, then $\Delta(a)_{e}<0$ and majority strategies dominate, so the minority share decreases. Thus, any deviation is self-correcting, and the equilibrium partition at $a^\star$ is restored.
\end{proof}
\begin{remark}
Therefore, the following can be generalized. There exists a minority bid set, and some minor bidders will occupy the set. This strategy of selecting the minority bid acts as an endogenous bid shading. It is also found that the minority strategy is efficient and adaptively stable.
\end{remark}
\section{Analysis: Minority Strategy Detection from Data}\label{sec:results}
\subsection{Summary Statistics}
The summary statistics Table \ref{tab:2} shows the interesting fact that the behavior of the data is not changing for imps\_hour, imps\_day, imps\_week, imps\_month, imps\_year, and imps\_alltime. Data points are the same. Skewness and Kurtosis are almost the same for all the variables. Variations are there because the density of data is not fixed. Therefore, it is clear that if we analyse imps\_hour data, then the results will also be true for all other variables and consistent as shown in Figures 1 to 6. Moreover,  this data is for a single bidder, so RTB is based on a 13-hour time bound. Therefore, the bidder is deciding for each hour to publish the ad. So, to understand the micro-level behavior, this will justify the purpose. 
\begin{table}[H]
	\centering
	\caption{Summary Statistics}
	\label{tab:2}
	\resizebox{\textwidth}{!}{%
		\begin{tabular}{|l|l|l|l|l|l|l|l|l|l|}
			\hline
			\textbf{Vars} & \textbf{n} & \textbf{Mean} & \textbf{SD} & \textbf{Median} & \textbf{Min} & \textbf{Max} & \textbf{Range} & \textbf{Skew} & \textbf{Kurtosis} \\ \hline
			\textbf{date}            & 1141470 & 43176.56  & 63.25      & 43168  & 43082.0   & 43313.0   & 231.0   & 0.29   & -1.02 \\ \hline
			\textbf{hour}            & 1141470 & 13.00     & 0.00       & 13     & 13.0      & 13.0      & 0.0     & NaN    & NaN   \\ \hline
			\textbf{adid\_anonymized} & 1141470 & 27.65     & 12.23      & 26     & 1.0       & 50.0      & 49.0    & 0.20   & -0.92 \\ \hline
			\textbf{bid}             & 1141470 & 25.00     & 14.43      & 25     & 0.1       & 49.9      & 49.8    & 0.00   & -1.20 \\ \hline
			\textbf{imps\_alltime}   & 1141470 & 33772.50  & 187526.20  & 29     & 0.0       & 9433550.0 & 9433550.0 & 20.42 & 732.23 \\ \hline
			\textbf{imps\_hour}      & 1141470 & 49.12     & 371.58     & 0      & 0.0       & 28672.0   & 28672.0 & 20.86 & 647.65 \\ \hline
			\textbf{imps\_day}       & 1141470 & 714.17    & 4905.63    & 0      & 0.0       & 318223.0  & 318223.0 & 20.32 & 634.86 \\ \hline
			\textbf{imps\_week}      & 1141470 & 3916.61   & 26178.94   & 0      & 0.0       & 1920525.0 & 1920525.0 & 22.69 & 868.64 \\ \hline
			\textbf{imps\_month}     & 1141470 & 14853.66  & 94691.09   & 2      & 0.0       & 5317739.0 & 5317739.0 & 21.68 & 777.94 \\ \hline
			\textbf{imps\_year}      & 1141470 & 33755.08  & 187529.12  & 14     & 0.0       & 9433550.0 & 9433550.0 & 20.42 & 732.19 \\ \hline
			\textbf{decoded\_date}   & 1141470 & NaN       & NA         & NA     & Inf       & -Inf      & -Inf     & NA    & NA    \\ \hline
		\end{tabular}%
	}
\end{table}
\subsection{Minority Strategy}
To test the simulation results proposed in Algorithms 1 and 2, this section identifies the minority strategy of the bidder in the original Yahoo Real Time Bidding data \cite{yahoo_auction_2020}. It identifies winning bids that result in higher impressions. Clustering: Group bids into "majority" (common) vs. "minority" (outliers).
\begin{table}[h!]
	\centering
	\label{tab:1}
	\begin{minipage}{0.45\linewidth}
		\centering
		\caption{Cluster Sizes}
		\begin{tabular}{|c|c|}
			\hline
			\textbf{Cluster} & \textbf{Size} \\ \hline
			0                & 571,888      \\ \hline
			1                & 569,582      \\ \hline
		\end{tabular}
	\end{minipage}%
	\hfill
	\begin{minipage}{0.5\linewidth}
		\centering
		\caption{Cluster Impression Summary}
		\begin{tabular}{|c|c|c|}
			\hline
			\textbf{Cluster} & \textbf{imps\_hour} & \textbf{imps\_day} \\ \hline
			0                & 74.248105           & 1143.596414        \\ \hline
			1                & 23.884666           & 282.998350         \\ \hline
		\end{tabular}
		\vspace{0.3cm}\\
		\textbf{Minority Cluster (Winning):} 0
	\end{minipage}
\end{table}
\begin{figure} [H]
	\centering
	\includegraphics[width=0.45 \textwidth]{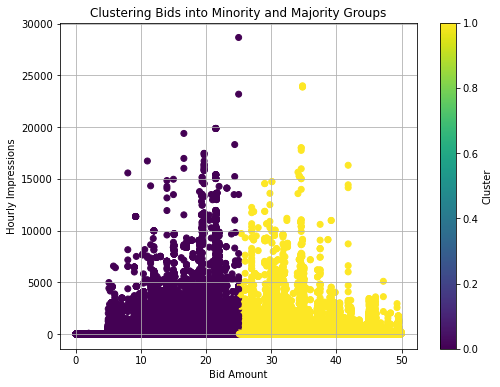}
	\caption{Minority Behavior Detection}
\end{figure} 
The cluster statistics are presented in Tables 3, 4, and 5, followed by Figure 11, representing the findings. The scatter plot reveals two distinct groupings.
\textit{X-axis:} Bid Amount,
\textit{Y-axis:} Hourly Impressions,
Color: Group classification (0 and 1) based on clustering (presumably KMeans).
Group Characteristics;
\textit{Left Group (Purple - Group 0):} Bid Amount: Lower range (0–25), Hourly Impressions: Low to moderate, Represents the "Minority Faction" with lower bids still gaining impressions, Participants in this group are cost-effective bidders, minimizing bids while maintaining competitiveness. \textit{Right Group (Yellow - Group 1):} Bid Amount: Higher range (25, 50), Hourly Impressions: Broad range, with more participants achieving higher; t participants bid more assertively, likely constituting the majority faction. Minority Faction Traits: Participants place lower bids but continue to receive impressions and secure impressions when competition is less intense (e.g., fewer high bidders); in Minority Game Theory, they strive to remain below the "majority threshold" (median bid) to optimize outcomes.\\
Notable Observations are: (i) Bid distribution for the minority faction is skewed towards lower bids, (ii) Concentration at low hourly impressions, indicating fewer impressions overall, (iii) Demonstrates adaptive strategies, (iv) Compromise on impressions to minimize expenses.
\subsection{Variance Scaling Analysis}
We will examine how the variance in bid amount changes relative to ad impressions.  Interpretation of Figure 12: The graph titled "Variance Scaling of Bid Amounts with Hourly Ad Impressions" plots the variance of bid amounts (on the y-axis) against bins of hourly ad impressions (x-axis).\\
Key Observations are as follows. \textit{Sharp Decrease in Variance (Bin 0 to Bin 1):} The variance of bid amounts starts very high ( $\sim$220) in the first bin (Bin 0). This indicates that bid amounts vary significantly when hourly ad impressions are in the lowest range (Bin 0). This could reflect more unpredictable or competitive bidding behavior when ad traffic is low.
\textit{Gradual Decline in Variance (Bin 1 to Bin 5):}
Moving to higher bins (Bins 1 to 5), the variance of bid amounts decreases steadily.
In the highest bin (Bin 5), the variance drops to around 80, indicating that the bid amounts become more stable and consistent when hourly ad impressions are higher.
Implications are as follows.\textit{Inverse Relationship:} There is an inverse relationship between the number of hourly ad impressions and the variance of bid amounts. As the number of hourly ad impressions increases, the variance in bids decreases.
\textit{Market Stabilization:} Higher hourly ad impressions may signify more mature or saturated markets, where bidding behavior is more predictable, leading to lower variance.
\textit{Competitive Bidding at Low Impressions:} The high variance at lower bins suggests significant variability in bidding behavior, possibly due to fewer advertisers or uncertainty in low-traffic periods.\\
The plot reveals that the variance of bid amounts decreases as the hourly ad impressions increase, highlighting more stable bidding patterns in higher impression ranges and greater unpredictability at lower ranges.
\begin{figure} [H]
	\centering
	\includegraphics[width=0.45 \textwidth]{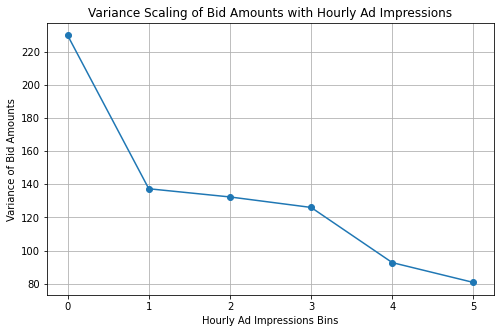}
	\caption{Variance Scaling Analysis}
\end{figure}
\section{Conclusion}
In first-price auction settings within Real-Time Bidding (RTB) environments, it is inherently challenging to identify the private valuation associated with each ad impression. The valuation is highly subjective and context-dependent, shaped by ad placement, viewer demographics, device characteristics, and auction timing. Consequently, the gap between a bidder's actual value and their submitted bid, which defines the extent of bid shading, remains analytically elusive when relying solely on aggregate bid data.\\
Leveraging the Yahoo Webscope RTB dataset, this study investigates this complexity through a data-driven and game-theoretic lens. Although the precise valuation-to-bid differential is unobservable, a robust empirical pattern emerges: bidders systematically adopt behaviors consistent with the Minority Game strategy. This pattern suggests that instead of pursuing uniform bidding behavior, advertisers strategically segment the temporal bid landscape and selectively engage in hourly impression zones where their presence constitutes a numerical minority. Such positioning enables them to avoid head-on competition and secure impressions at lower marginal costs.\\
This observed behavior reflects the dual-layered heterogeneity inherent in the RTB ecosystem product heterogeneity (i.e., variance in impression attributes) and bidder heterogeneity (i.e., variance in strategy, valuation, and objective functions). Given these structural dynamics, a uniform equilibrium bidding strategy across all bidders and time periods is implausible. Instead, this study demonstrates, both theoretically and empirically, that adaptive minority strategies provide a viable explanation for the observed bid distributions.\\
By integrating Minority Game Theory with temporal clustering and impression-level analytics, the paper offers novel insights into strategic behavior under uncertainty in decentralized auction markets. The findings contribute to the literature on auction dynamics and computational advertising and have practical implications for designing more adaptive, fairness-aware bidding algorithms in future RTB systems.\\
\indent For further study, this cue will serve as an important methodology where the ad exchange can set the strategy by creating an artificial market, assuming that each bidder is selecting the minority strategy to match the demand and supply side platforms. 
\section{Practical and Theoretical Implications}
Canonical auction models for a single, indivisible good, typically under independent private values, assume that equilibrium bids are strictly increasing in valuations, and that greater competition (proxied by the number of bidders) raises expected prices and, in truthful mechanisms (such as second-price formats), supports value revelation. This study's evidence from real-time bidding (RTB) ad exchanges departs from these benchmark predictions. The empirical patterns and the proposed model indicate that round-wise winners frequently adopt a \emph{minority bidding} posture, placing bids below the dominant cluster consistent with deliberate bid shading in repeated, budget- and pacing-constrained, first-price environments. This behavior attenuates the standard link between bidder count and truthful revelation and can depress platform revenue when mechanism choice and pricing rules are misaligned with RTB frictions.\\
Practically, these findings motivate mechanism and policy adjustments at the exchange level. First, platform design should explicitly anticipate strategic shading by monitoring the prevalence of minority bid wins and by calibrating dynamic reserve prices or floors to counter systematic underbidding without deterring participation. Second, scoring and allocation rules can be stress-tested against bidder heterogeneity (budgets, frequency caps, learning dynamics) to ensure that price discovery remains informative. Third, where market conditions permit, hybrid formats (e.g., guarded second-price with floors or outcome-based pricing) may reduce incentives for coordinated shading while preserving allocative efficiency. Finally, valuation of digital inventory should be informed by models that account for repeated interaction and strategic bid adjustment, so that competition remains intense in practice rather than only in theory.
\section{Limitations}
The analysis relies on the Yahoo Webscope RTB dataset and a specific modeling framework. Generalization should be assessed with additional datasets and market contexts (e.g., different exchanges, verticals, and auction rules). Future research could test the minority-bidding mechanism across alternative RTB environments, incorporate richer covariates on budgets and pacing, and examine robustness under varying reserve policies and format changes.
\section*{Conflict of Interest}
I am certifying that we have no affiliations with or involvement in any organization or entity with any financial interest (such as honorarium; educational grants; participation in speakers’ bureaus; membership, employment, consultancies, stock ownership, or other equity interest; and expert testimony or patent-licensing arrangements), or non-financial interest (such as personal or professional relationships, affiliations, knowledge or beliefs) in the subject matter or materials discussed in this manuscript. \\
\section*{Data availability statement}
Data subject to third-party restrictions:\\
The data that support the findings of this study are available from [Brendan Kitts (2020), Auction State for a Sample of Real-Time Bid Video Ads Version 1.0, Yahoo Research Webscope dataset ydata-o\&o-video-auction-landscapes-2018-v1\_0, https://webscope.sandbox.yahoo.com/, June 20, 2020]. Restrictions apply to the availability of these data, which were used under license for this study. The author's Data is available, with permission from "Brendan Kitts (2020), Auction State for a Sample of Real-Time Bid Video Ads Version 1.0, Yahoo Research Webscope dataset".
\section*{Declaration of Funding}
No funding was received.
\bibliographystyle{plainnat}
\bibliography{myreferences}
\end{document}